\documentclass[runningheads]{llncs}

\usepackage[T1]{fontenc}
\usepackage{mathtools}
\usepackage{graphicx}
\usepackage{subcaption}
\usepackage{microtype}
\usepackage{fnpct}
\usepackage{cite}
\usepackage{xspace}
\usepackage{todonotes}
\usepackage{amsmath,amssymb}
\usepackage{thm-restate}
\usepackage{tikz}

\usepackage{lineno}

\usepackage{fancybox}
\cornersize*{10pt}
\newcommand{\myproblem}[3]{
  \par
\begin{description}
      \item[Problem] #1
      \item[Input] #2
      \item[Question] #3
    \end{description}
\par
}

\usepackage{hyperref}
\usepackage{cleveref}

\let\doendproof\endproof
\renewcommand{\endproof}{\hfill$\qed$\doendproof}

\graphicspath{{./figures/}}

\newcommand{\pqprob}{\textsc{Sequential PQ-Or\-der\-ing}\xspace}
\newcommand{\simpq}{\textsc{Simultaneous PQ-Ordering}\xspace}

\newcommand{\tlp}{\textsc{proper $\mathcal T$-Level Planarity}\xspace}

\newcommand{\hvprob}{\textsc{Semi-Fixed HV-Segment Intersection Graph Recognition}\xspace}
\newcommand{\hvp}{\textsc{SF-HV-SEG}\xspace}

\title{Segment Intersection Representations, Level Planarity and Constrained Ordering Problems}
\titlerunning{HV-Segment Intersection Representations}

\author{
  Simon D.~ Fink\inst{1}\orcidID{0000-0002-2754-1195} \and
Matthias Pfretzschner\inst{2}\orcidID{0000-0002-5378-1694} \and
Peter Stumpf\inst{3,4}\orcidID{0000-0003-0531-9769}
}

\institute{
  Algorithms and Complexity Group, Technische Universität Wien, Austria
  \email{sfink@ac.tuwien.ac.at} \and
Faculty of Computer Science and Mathematics, University of Passau, Germany
  \email{pfretzschner@fim.uni-passau.de} \and
  Department of Theoretical Computer Science, Faculty of Information Technology, Czech Technical University in Prague, Czech Republic \and
  Department of Applied Mathematics, Charles University Prague, Czech Republic
  \email{stumpf@kam.mff.cuni.cz}
}

\authorrunning{S.~D.~Fink, M.~Pfretzschner, and P.~Stumpf}

\begin{document}
\maketitle

\begin{abstract}
  In the Segment Intersection Graph Representation Problem, we want to represent the vertices of a graph as straight line segments in the plane such that two segments cross if and only if there is an edge between the corresponding vertices.
  This problem is NP-hard (even $\exists\mathbb{R}$-complete \cite{sch-cos-10}) in the general case \cite{km-igo-94} and remains so if we restrict the segments to be axis-aligned, i.e., horizontal and vertical \cite{kra-asp-94}.
  A long standing open question for the latter variant is its complexity when the order of segments along one axis (say the vertical order of horizontal segments) is already given \cite{nf-ppa-92,kra-asp-94}.

  We resolve this question by giving efficient solutions using two very different approaches that are interesting on their own.
  First, using a graph-drawing perspective, we relate the problem to a variant of the well-known Level Planarity problem, where vertices have to lie on pre-assigned horizontal levels.
  In our case, each level also carries consecutivity constraints on its vertices; this Level Planarity variant is known to have a quadratic solution.

  Second, we use an entirely combinatorial approach and show that both problems can equivalently be formulated as a linear ordering problem subject to certain consecutivity constraints.
  While the complexity of such problems varies greatly, we show that in this case the constraints are well-structured in a way that allows a direct quadratic solution.
  Thus, we obtain three different-but-equivalent perspectives on this problem: the initial geometric one, one from planar graph drawing and a purely combinatorial one.

\end{abstract}

\section{Introduction}
In a seminal work from 1989, Kratochv{\'\i}l and Matou{\v s}ek initiated the systematic study of intersection graphs of segments forming the class \textsc{SEG} \cite{km-igo-94}.
A graph is in \textsc{SEG} if its vertices can be represented through straight line segments in the plane that intersect if and only if the corresponding vertices share an edge.
They not only showed that testing membership in \textsc{SEG} is NP-hard, but also through showing that a representation of a graph in \textsc{SEG} might require coordinates with exponential precision \cite{km-igo-94}, laid the foundation for the study of a new complexity class denoted as $\exists\mathbb{R}$ \cite{sch-cos-10}.
They studied further variants of \textsc{SEG} (see Fig.~1 of \cite{km-igo-94}), in one direction further allowing the segments to have a bounded number of bends \cite{kra-sgi-91,km-nhr-89}, and in the other direction restricting the segments to a bounded number of slopes \cite{kra-asp-94}.
In the latter setting, even the case where only two directions, say horizontal and vertical, are allowed (in their work denoted as \textsc{2-DIR}, or \textsc{PURE-2-DIR} also forbidding overlaps) turned out NP-complete \cite{kra-asp-94}.

Kratochv{\'\i}l and Neŝetril observed that further prescribing the orders of horizontal and vertical segments, separately along both axes, finally turned the problem tractable \cite{kra-asp-94,nf-ppa-92}.
They did so through two different equivalent reformulations of the problem.
First note that, in the case without overlaps, we are given a bipartite graph $G=(H\cup V, E)$ together with orders $\sigma_H$ and $\sigma_V$ of $H$ and $V$, respectively; see \Cref{fig:example} for an example instance.
In their first formulation, they used a 2-layer drawing of $G$ with $H$ and $V$ respectively on two parallel lines, using $\sigma_H$ as order on the one layer and $\sigma_V$ on the other; see \Cref{fig:volkswagen} and \cite[Fig.~19]{kra-asp-94}.
They observed that an instance is positive if and only if it avoided a certain configuration of edges in this drawing, which they called ``volkswagen''.
Alternatively, following an approach of Hartman, Newman and Ziv \cite{hnz-ogi-91}, the yes-instances can also be classified via their adjacency matrix, which shall not contain a configuration called ``cross'', see \Cref{fig:matrix-cross} and \cite[Proposition 3.1]{hnz-ogi-91}.

\begin{figure}[t]
  \begin{subfigure}{.15\textwidth}
    \includegraphics[scale=0.8,page=2,trim={20pt 0 0 0},clip]{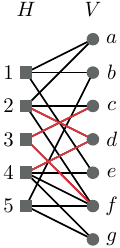}
    \subcaption{}
    \label{fig:volkswagen}
  \end{subfigure}\begin{subfigure}{.48\textwidth}
    \centering
    \includegraphics[scale=0.8,page=5]{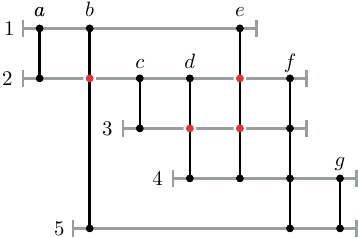}
    \subcaption{}
    \label{fig:segments-alpha}
  \end{subfigure}\begin{subfigure}{.37\textwidth}
    \includegraphics[scale=0.8,page=6]{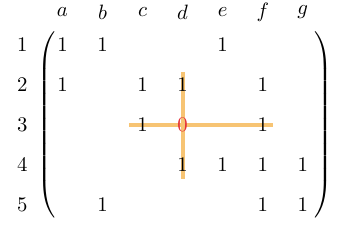}
    \subcaption{}
    \label{fig:matrix-cross}
  \end{subfigure}
  \caption{
    \textbf{(a)} A 2-layer drawing of a bipartite graph $G=(H\cup V, E)$.
      Highlighted in red is one ``volkswagen'' configuration present due to the orders used for $H$ and $V$ as well as the missing $3d$ edge.
    \textbf{(b)} An invalid HV-segment representation using the same vertex orders, with unwanted intersections highlighted in red.
    \textbf{(c)} The adjacency matrix of $G$, with an unwanted ``cross'' configuration centered on the absent $3d$ edge highlighted in orange. Empty cells correspond to a value of 0.
  }
  \label{fig:example}
\end{figure}
  
With the complexity of recognizing segment intersection graphs thus rather well established, Kratochv{\'\i}l and Neŝetril posed the following open problem \cite{kra-asp-94,nf-ppa-92}:
what if only one of the two orders (e.g. the vertical order of horizontal segments) is fixed?
In the above equivalent formulations, this is equivalent to either allowing to permute the vertices along one of the two parallel lines, or to permute the matrix along one axis (e.g. permuting its columns).
Formally, this can be stated as follows; see also \Cref{fig:reduction}.
\myproblem{\hvprob}{
  A bipartite graph $G = (H \cup V, E)$ and a linear order $\sigma_H$ of $H$.
}{
  Is $G$ the intersection graph of horizontal segments $H$ and vertical segments $V$, where the horizontal segments have the fixed vertical order~$\sigma_H$?
}

In this work, we will finally resolve this question after roughly 30 years, using two further entirely different, yet equivalent perspectives on the problem.
Akin to the visual ``volkswagen'' classification, our first perspective will be formulated in terms of Graph Drawing using a problem called \textsc{Level Planarity}.
In \Cref{sec:hv}, we will show that \hvprob (or \hvp for short) can be reduced to a variant of \textsc{Level Planarity}, thus also allowing us to use its known quadratic-time solution.
In \Cref{sec:hv-to-level} we will conversely show that we can also reduce the problems in the other direction, showing their equivalence up to (in)efficiencies in representing problem instances.
We will turn to our second perspective in \Cref{sec:seqpq}, where we give an entirely combinatorial formulation of \hvp.
We again show equivalence for this formulation and also show it can directly be solved in quadratic time without the need for further reductions.

\section{Semi-Fixed HV-Segment Intersection Graphs}\label{sec:hv}

\begin{figure}[t]
  \hspace{-.02\textwidth}
  \begin{subfigure}{.18\textwidth}
    \includegraphics[scale=0.75,page=1]{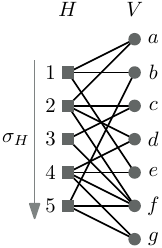}
    \subcaption{}
    \label{fig:volkswagen2}
  \end{subfigure}\begin{subfigure}{.4\textwidth}
    \centering
    \includegraphics[scale=0.75,page=3]{segments}
    \subcaption{}
    \label{fig:segments-fixed}
  \end{subfigure}\begin{subfigure}{.4\textwidth}
    \includegraphics[scale=0.75,page=4]{segments}
    \subcaption{}
    \label{fig:tlp-inst}
  \end{subfigure}
  \caption{
    \textbf{(a)} The \hvp instance from \Cref{fig:volkswagen}, only fixing the order $\sigma_H$.
    \textbf{(b)} A corresponding HV-segment representation with a different order of $V$ than in \Cref{fig:segments-alpha}, showing that it is a yes-instance.
    \textbf{(c)} An equivalent instance of \tlp. The consecutive sets of every level are highlighted in orange.
  }
  \label{fig:reduction}
\end{figure}

In this section, we show that \hvp can be solved in quartic time using a reduction to some variant of the \textsc{Level Planarity} problem.
For \textsc{Level Planarity}, we are given a graph together with prescribed y-coordinates for its vertices (also called \emph{levels}), and seek a crossing-free drawing that respects these y-coordinates and also has all edges y-monotone.
This problem has been thoroughly researched in the field of Graph Drawing, e.g. in terms of efficient and also simple solutions \cite{jlm-lpt-98,jl-lpe-02,fpss-htm-13,brs-lpt-18,brs-lpt-22} and also its relationship to other planar graph drawing problems \cite{sch-tat-13,albfr-ibpcl-15}.
We call the tuple $(G,\ell)$ consisting of a graph $G=(V,E)$ and a \emph{leveling function} $\ell : V \to \mathbb{N}$ a \emph{level graph}.
A level graph (and especially its leveling function) is called \emph{proper} if all its edges have their endpoints on adjacent levels, i.e., if we have $\ell(v) = \ell(u) + 1$ for all $uv\in E$.
Similarly, a path $v_1, v_2, \dots, v_k$ is \emph{level-monotone} if $\ell(v_{i+1}) = \ell(v_i) + 1$ for all $i \in [k-1]$ with $[k-1]\coloneq\{0,\ldots,k-1\}$.

In the context of our work, we will use a further constrained variant that also allows us to require sets of vertices on the same layer to be \emph{consecutive}, i.e., ordered such that no other vertices of that level lie between them.
For now, we require the sets of consecutive vertices to form a laminar family, i.e., any two sets are either disjoint or one contains the other.
This problem is then known to have a quadratic-time solution if the instance is proper, while it is NP-complete on non-proper instances \cite{albfr-ibpcl-15,albfr-ibpcl-14}.
Formally, the former, tractable case of the problem is defined as follows.

\myproblem{\tlp}{
  A graph $G=(V,E)$, a proper leveling function $\ell : V \to \mathbb{N}$, and for each level $i$ a laminar family $T_i$ of consecutivity constraints on $\ell^{-1}(i)$\footnote{Using the containment-relationship of laminar families allows us to equivalently represent $T_i$ as a tree of linear size, thus obtaining the original formulation of the problem \cite{albfr-ibpcl-15}.}.
}{
  Is there a planar drawing of $G$ where each vertex $v\in V$ has y-coordinate $\ell(v)$,
  all edges are drawn y-monotone, and
  the horizontal order of vertices on each level $i$ satisfies the consecutivity constraints of~$T_i$?
}

To now reduce \hvp to this problem, we create a level for each horizontal segment along the given order; see \Cref{fig:reduction}.
For each vertical segment, we create an edge that spans from the first to the last level of horizontal segments that shall be intersected.
We subdivide the edges to make the instance proper.
The crucial insight to \hvp is now that all vertical segments $\mathcal I(i)$ that intersect a given horizontal segment $h_i$ need to be consecutive on the horizontal line corresponding to $h_i$.
This means that any other vertical segment that spans, but does not intersect $h_i$, needs to be placed before or after~$\mathcal I(i)$.
This can easily be modeled using the per-level consecutivity constraints of \tlp.

\begin{theorem}
  \label{thm:hv-to-rtlp}
  There exists a quadratic-time reduction from \hvp to \tlp where the input graph is the disjoint union of level-monotone paths.
\end{theorem}
\begin{proof}
  Let $((H \cup V, E), \sigma_H)$ be an instance of \hvp.
  For $k \coloneqq |H|$, let $\ell: H \to [k]$ be the bijection that maps every horizontal segment in $H$ to its corresponding position in the linear order defined by $\sigma_H$.
  We call $\ell(h)$ the \emph{level} of $h \in H$.
   
  For every vertical segment $v \in V$, we let the interval $L(v) = [\min\{\ell(h) : vh \in E\}, \max\{\ell(h) : vh \in E\}] \subseteq [k]$ denote the \emph{lifespan} of $v$, that is the interval spanned by the lowest and highest levels of vertices adjacent to $v$.
  For every $i \in [k]$ and its corresponding horizontal segment $h = \ell^{-1}(i)$, the set $\mathcal A(i) \coloneqq \{v \in V : i \in L(v)\}$ denotes the \emph{active segments} of level $i$ and the set $\mathcal I(i) = \{v \in V : vh \in E\} \subseteq \mathcal A(i)$ describes the \emph{intersecting segments} of level $i$.
  
  We construct an equivalent \tlp instance $(G, \gamma,$ $(T_1, \dots, T_k))$ with $k$ levels as follows.
  The horizontal segments $H$ correspond bijectively to the levels of $(G, \gamma)$ as defined by $\ell$.
  Every vertical segment $v \in V$ is represented by a level-monotone path in $G=(V,E)$ that contains a vertex on every level that $v$ is active in, i.e., $V \coloneqq \{v_i : i \in [k], v \in \mathcal A(i)\}$ with $\gamma(v_i) \coloneqq i$ and $E \coloneqq \{v_iv_{i+1} : i \in [k-1], v \in \mathcal A(i) \cap \mathcal A(i+1)\}$.
  Finally, on every level $i \in [k]$, we constrain the vertices corresponding to the intersecting segments $\mathcal I(i)$ to be consecutive on that level, that is, we set $T_i \coloneqq \{\{v_i : v \in \mathcal I(i)\}\}$.
  In particular, observe that we obtain only a single consecutivity constraint per level.
  
  Note that the resulting level graph $(G, \gamma)$ is proper and, in the worst case, the number of vertices in $G$ is quadratic in $|H \cup V|$.
  We show that $((H \cup V, E), \sigma_H)$ is a yes-instance of \hvp if and only if $(G, \gamma, (T_1, \dots, T_k))$ is a yes-instance of \tlp.
  
  First assume that $(G, \gamma, (T_1, \dots, T_k))$ is a yes-instance of \tlp and let $\Gamma$ be a corresponding drawing of $G$.
  Since $G$ is the disjoint union of monotone paths and $\Gamma$ is planar, we can alter $\Gamma$ such that each path is represented by a straight vertical line, without changing the order of vertices on any level.
  We obtain an HV-segment representation of $((H \cup V, E), \sigma_H)$ by replacing each of these paths with a vertical segment, and subsequently adding on every level $i$ a horizontal segment that intersects the vertices of $T_i$.
  Since these vertices are consecutive on level $i$ in $\Gamma$, the horizontal segment intersects precisely the vertical segments it is adjacent to in $(H \cup V, E)$.
  
  Conversely, assume that $((H \cup V, E), \sigma_H)$ is a yes-instance of \hvp and assume that a corresponding HV-segment representation $\Delta$ is given.
  We obtain a level-planar drawing $\Gamma$ of $(G, \gamma)$ by simply replacing every vertical segment with a path that contains a vertex on every level where it is active.
  Now assume that, for some level $i \in [k]$, the single consecutivity constraint of $T_i$ is not satisfied in $\Gamma$.
  This implies that the intersecting segments $\mathcal I(i)$ of level~$i$ are not consecutive with respect to all active segments $\mathcal A(i)$ of level $i$ in $\Delta$.
  Therefore, a vertical segment $v \in \mathcal A(i) \setminus \mathcal I(i)$ lies between two vertical segments of $\mathcal I(i)$ on level~$i$.
  Since the horizontal segment $\ell^{-1}(i)$ intersects exactly the segments in $\mathcal I(i)$ and $v \notin \mathcal I(i)$, $i$ cannot be part of the lifespan of $v$.
  But this is a contradiction, since~$v \in \mathcal A(i)$.
  Therefore, $(G, \gamma, (T_1, \dots, T_k))$ is a yes-instance of \tlp.
\end{proof}

Since \tlp admits a quadratic-time algorithm \cite{albfr-ibpcl-14,albfr-ibpcl-15}, we can thereby solve the initial instance in quartic time overall.

\begin{corollary}\label{cor:hvp-quartic}
  \hvp can be solved in $O(n^4)$ time.
\end{corollary}

\section{Relation to Level-Planarity Variants}\label{sec:hv-to-level}
In this section, we show that \hvp is quadratic-time equivalent to \tlp.
Since \hvp reduces to \tlp in quadratic time by \Cref{thm:hv-to-rtlp}, it remains to show the other direction.
To this end, we first show that every instance of \tlp can be reduced in linear time to an equivalent instance where the graph is a matching.
The edges of this matching then correspond to the vertical segments of the equivalent \hvp instance, while each consecutivity constraint results in a horizontal segment.
Since each level of the \tlp instance may have a linear number of consecutivity constraints, this yields a quadratic-time reduction.

\begin{lemma}
  \label{lem:level-monotone}
  Given an instance of \tlp, we can find in linear time an equivalent instance where the graph is a matching.
\end{lemma}
\begin{proof}
  We assume without loss of generality that the input graph contains no isolated vertices.
  Given an instance of \tlp, first ``explode'' every vertex~$v$, i.e., replace $v$ with $\deg(v)$ degree-1 vertices on the same level, adjacent to the original neighbors of $v$.
  Subsequently, add an additional consecutivity constraint to the level of $v$ that ensures that the new degree-1 vertices that stem from $v$ are consecutive.
  Due to this consecutivity, a corresponding drawing of the original graph can be obtained by simply merging the new degree-1 vertices back into a single vertex.
  Conversely, to obtain a drawing of the new instance, perform the explosions for each vertex and order the new degree-1 vertices on the same level according to the order of edges around the original vertex.
  
  Observe that the resulting family of consecutivity constraints is still laminar and since every vertex now has degree 1, the graph is a matching.
  Moreover, the number of new vertices and the combined size of the new consecutivity constraints each correspond to the sum of vertex-degrees in the original instance.
  Therefore, the reduction is linear.
\end{proof}

\begin{lemma}
  \label{lem:tlp-hv}
  There exists a quadratic-time reduction from \tlp to \hvp.
\end{lemma}
\begin{proof}
Due to \Cref{lem:level-monotone}, we can assume that the input graph is a matching.
  We first show that we can compute an equivalent instance of \tlp with only one consecutivity constraint per level.
  For each level $i$, let $T_i$ denote the corresponding laminar family with $|T_i|$ consecutivity constraints.
  We split $i$ into $|T_i|$ levels $i^1, \dots, i^{|T_i|}$ and redistribute its constraints as follows.
  Replace each vertex $v$ on $\ell$ by a path of $|T_i|$ vertices $v^1, \dots v^{|T_i|}$ on levels $i^1, \dots, i^{|T_i|}$, respectively.
  Since the input graph is a matching, $v$ has exactly one neighbor.
If this neighbor lies on a lower level, we assign the corresponding edge to $v^1$, otherwise we assign it to $v^{|T_i|}$.
  Observe that due to these parallel paths, the vertices on all $|T_i|$ levels have the same order in a level planar drawing.
  This allows us to distribute $T_i$ over the levels $i^1, \dots, i^{|T_i|}$, leaving one constraint per level.
  As each level $i$ of the input instance may have a linear number of constraints, this yields a quadratic-time reduction.
  Observe that the resulting instance is the disjoint union of level-monotone paths and each level has exactly one constraint.
  
  Now consider an instance $I = (G, \ell, (T_1, \dots, T_k))$ of \tlp where $G$ consists of disjoint level-monotone paths and that has exactly one consecutivity constraint per level, i.e., $|T_i| = 1$ for all $i\in[k]$.
  We construct an equivalent instance $I' = ((H \cup V, E), \sigma_H)$ as follows.
  The (vertices represented by) vertical segments $V$ correspond bijectively to the level-monotone paths of $G$, while the (vertices represented by) horizontal segments $H$ correspond bijectively to the levels of $I$.
  The order $\sigma_H$ orders the horizontal segments according to the order of their corresponding levels.
For each level $i$, the horizontal segment corresponding to level $i$ is adjacent to the vertical segments of $V$ whose corresponding paths in~$I$ have a vertex in $T_i$; see \Cref{fig:reduction}.
  
  Note that, as in the proof of \Cref{thm:hv-to-rtlp}, the level-monotone paths of $I$ correspond bijectively to the vertical segments in $I'$, while the one-per-level consecutivity constraints correspond to the horizontal segments.
  The equivalence of the two instances can thus be shown analogously as done in the proof of \Cref{thm:hv-to-rtlp}.
\end{proof}

This now allows us to reduce between \hvp and \tlp with a quadratic factor that is mostly due to differences in how efficient the two problems represent their input instances.
That is, the latter one more efficiently represents multiple consecutivity constraints per level than the former, using the representation of a laminar family of sets as tree.
In contrast, it also on each level needs to make explicit those elements that shall lie outside of the consecutive subset, which is implicit in the former representation.

\section{Relation to Constrained Ordering Problems}\label{sec:seqpq}
The reduction in \Cref{thm:hv-to-rtlp} already hints that \hvp is essentially a constrained ordering problem: we seek a linear order of the vertical segments, such that, on each of the fixed horizontal layers, the active segments satisfy a prescribed consecutivity constraint.
Finding a linear order of a ground set $X$ subject to some constraints that require a subset $A$ to appear consecutively, that is, uninterrupted with regards to some other set $B$, is indeed a common combinatorial problem.
This problem is in general NP-complete for arbitrary $A$ and $B$, due to a straight-forward reduction from the well-known \textsc{Betweenness} problem \cite{opa-top-79}.
In contrast it is, e.g., linear-time solvable if we restrict $B$ to be the complement of $A$ with regards to $X$ (using solely the PQ-tree data structure described further below \cite{bl-tcopi-76}).
In this section, we want to show that \hvp and \tlp can equivalently be seen as a relaxed variant of this latter setting, where all given $(A,B)$ consecutivity constraints can be ordered such that every element of $X$ appears in a consecutive subsequence of constraints.

\begin{figure}[t]
  \begin{subfigure}{.45\textwidth}
    \centering
    \includegraphics[scale=1,page=1]{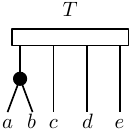}
    \subcaption{}
    \label{fig:pqtree}
  \end{subfigure}
  \begin{subfigure}{.55\textwidth}
    \centering
    \includegraphics[scale=1,page=2]{laminarPQTrees}
    \subcaption{}
    \label{fig:laminar}
  \end{subfigure}
  \caption{
    \textbf{(a)} A PQ-tree $T$ with $L(T) = \{a,b,c,d,e\}$ and consecutivity constraints $\{a,b\}, \{a,b,c\}, \{c, d\}$, and $\{d, e\}$. P-nodes and C-nodes are represented by black circles and white rectangles, respectively.
    \textbf{(b)} Two PQ-trees $T^1$ and $T^2$ with the laminar consecutivity constraints $\{a,b\}, \{a,b,c\}, \{d,e\}$ and $\{c, d\}$, respectively.
  }
  \label{fig:pq-trees}
\end{figure}

Recall that a laminar family of consecutivity constraints can succinctly be represented as a tree with the constrained elements as leaves.
If we now drop the restriction that the constraints form a laminar family, i.e., allow them to have arbitrary intersections, we can still represent them efficiently using a data structure called PQ-tree~\cite{bl-tcopi-76}.
A \emph{PQ-tree} $T$ is a rooted tree that represents linear orders of its leaves $L(T)$, subject to a set of constraints that require certain leaves to be consecutive; see \Cref{fig:pqtree} for an example.
Each of its inner nodes is either a \emph{P-node} or a \emph{Q-node}, where the former allow their children to be permuted arbitrarily, while the latter fix the order of their children up to reversal.
The linear orders of $L(T)$ that can be obtained in this way are the orders that are \emph{represented} by $T$.
Conversely, we say that a linear order $\sigma$ of a set $X \supseteq L(T)$
\emph{satisfies} $T$ if $\sigma$ restricted to $L(T)$ is represented by $T$.
For a given (not necessarily laminar) family of consecutivity constraints on a ground set~$X$, the PQ-tree $T$ with $L(T)=X$ that represents exactly those orders that satisfy all constraints can be found in linear time~\cite{bl-tcopi-76}.
If there exists no order of $L(T)$ that satisfies all consecutivity constraints, then $T$ is the \emph{null tree}.
We now use this data structure to formulate the following constrained ordering problem and, in the later parts of this section, show it to be the combinatorial analogue of the problems from the previous sections.

\myproblem{\pqprob}{
  A ground set $X$ and a sequence $T_1, \dots, T_k$ of PQ-trees with $X = {L(T_1) \cup \dots \cup L(T_k)}$ such that, for every $x \in X$, the PQ-trees that contain $x$ are consecutive in the sequence $T_1, \dots, T_k$.
}{
  Is there a linear order of $X$ that satisfies all PQ-trees $T_1, \dots, T_k$?
}

We remark that, when we drop the requirement that each element of the ground set must appear in a consecutive subset of the PQ-tree sequence, the problem becomes NP-complete.\footnote{In a sense, this is analogous to dropping the ``proper'' part of \tlp, where the then-allowed long edges do not appear in the trees of levels they span.}
This can again easily be proven via reduction from the \textsc{Betweenness} problem by representing each betweenness constraint individually by a PQ-tree consisting of a single Q-node.
However, the following lemma shows that thanks to the additional requirement in \pqprob, it suffices to ensure that common elements are ordered consistently for PQ-trees that are adjacent in the sequence.

\begin{lemma}
  \label{lem:pqprobCharacterization}
  Let $\mathcal I = (X, (T_1, \dots, T_k))$ be an instance of \pqprob.
  Then $\mathcal I$ is a yes-instance if and only if there exist linear orders $\sigma_1, \dots, \sigma_k$ of $L(T_1), \dots, L(T_k)$ such that
  \begin{enumerate}
    \item for every $i \in [k]$, $\sigma_i$ satisfies $T_i$ and
    \item for every $i \in [k-1]$ and every $x, y \in L(T_i) \cap L(T_{i+1})$ it is $x \leq_{\sigma_{i}} y \Leftrightarrow x \leq_{\sigma_{i+1}} y$
  \end{enumerate}
\end{lemma}
\begin{proof}
  It is clear that both properties hold if $\mathcal I$ is a yes-instance.
  Conversely, assume there exist linear orders $\sigma_1, \dots, \sigma_k$
  of $L(T_1), \dots, L(T_k)$ that satisfy both properties.
We show that $X$ admits a linear order $\sigma$ that extends $\sigma_1, \dots, \sigma_k$ by induction over the number of PQ-trees.
  The statement clearly holds for $k = 1$.
  Assume $k>1$ and that the statement holds for $\mathcal I' = (X', (T_1,\dots,
  T_{k-1}))$ with $X' = {L(T_1) \cup \dots \cup L(T_{k-1})}$, i.e., there is a linear order $\sigma'$ of $X'$ that extends
  $\sigma_1, \dots, \sigma_{k-1}$.
  We construct a linear order $\sigma$ of $X$ that extends $\sigma_1, \dots, \sigma_k$ as follows. 
  By Property~2, $\sigma_{k-1}$ and $\sigma_k$ (and thus $\sigma'$ and $\sigma_k$) induce the same linear order
  $\tau$ on $X'\cap L(T_{k}) = L(T_{k-1})\cap L(T_k)$.
  Consider $\sigma'\cup \sigma_k$ as the edges of a directed graph on $X$ and
  observe that it is acyclic since $\sigma'$, $\sigma_k$ and~$\tau$ are linear
  orders. Hence, we can choose any topological ordering of the graph as~$\sigma$.
  By construction $\sigma$ extends $\sigma_1, \dots, \sigma_k$ and thus satisfies $T_1, \dots T_k$.
  This means that~$\mathcal I$ is a yes-instance of \pqprob.  
\end{proof}

We now will use a known similar characterization for proper level graphs to establish the equivalence of \pqprob and \tlp.
It was initially given in \cite{rsb-asf-01} and states that a straight-line drawing is level planar if and only if the endpoints of every pair of edges between adjacent levels are ordered consistently.
\begin{lemma}[\citen{rsb-asf-01}]
  \label{lem:levelPlanarEmbedding}
  A proper level graph $(G, \ell)$ on $k$ levels is level planar if and only if there exists linear orders $\sigma_1,\ldots,\sigma_k$ of the vertices on individual levels $\ell^{-1}(1),\ldots,\ell^{-1}(k)$, respectively, that satisfy the following property.
  For every pair $(a, b), (u, v)$ of disjoint edges with $i = \ell(a) = \ell(u) = 1 + \ell(b) = 1 + \ell(v)$, we have $a <_{\sigma_i} u \Leftrightarrow b <_{\sigma_{i+1}} v$.
\end{lemma}

Using \Cref{lem:pqprobCharacterization,lem:levelPlanarEmbedding}, we can now show that \pqprob and \tlp are linear-time equivalent.
\begin{lemma}\label{lem:pqprob2tlp}
  There exists a linear-time reduction from \pqprob to \tlp.
\end{lemma}
\begin{proof}
Recall that \tlp requires a laminar set of consecutivity constraints on each level.
  We therefore first show how we can translate each PQ-tree into an equivalent set of PQ-trees that represent laminar consecutivity constraints.
  
  Each P-node corresponds to a constraint that requires consecutivity of the leaves in its subtree.
  Similarly, a Q-node can be formulated as multiple consecutivity constraints:
  for each pair of adjacent children of a Q-node, the union of the leaves of their subtrees must be consecutive \cite{boo-pta-75}.
  This allows us to translate each PQ-tree $T$ into two equivalent PQ-trees $T^1$ and $T^2$ that each represent laminar families of consecutivity constraints as follows.
  Assign all constraints that stem from P-nodes to $T^1$, and distribute the constraints that stem from pairs of children of Q-nodes alternatingly between $T^1$ and $T^2$; see \Cref{fig:laminar} for an example.
  Alternating between the two trees ensures that two constraints in the same tree do not intersect if one is not contained in the other.
  The PQ-trees $T^1$ and $T^2$ together contain precisely the consecutivity constraints of $T$ and can be computed in linear time. Given an instance of \pqprob, we can therefore compute in linear time an equivalent instance $\mathcal I = (X, (T_1, \dots, T_k))$ where each $T_i$ represents a laminar set of consecutivity constraints.
  
  We now construct an equivalent instance of \tlp as follows.
  We obtain the level graph $(G = (V, E), \ell)$ by introducing a level for each $i \in [k]$, which corresponds to the PQ-tree $T_i$.
  Each level $i$ contains a vertex for every element of $L(T_i)$, i.e., $V \coloneqq \{x_i : i \in [k], x \in L(T_i)\}$ and $\ell(x_i) = i$.
  Moreover, we connect vertices on adjacent levels that correspond to the same element, i.e., $E \coloneqq \{x_ix_{i+1} : i \in [k-1], x \in L(T_i) \cap L(T_{i+1}) \}$.
Finally, we annotate each level $i$ with the laminar family $T'_i$ of consecutivity constraints that is obtained from $T_i$ by replacing every leaf $x \in L(T_i)$ with the corresponding vertex $x_i$, resulting in the instance $\mathcal I' = (G, \ell, (T'_1, \dots, T'_k))$.
  
  For each level $i \in [k]$, the linear order of the vertices on level $i$ corresponds to the linear order $\sigma_i$ that a linear order $\sigma$ of $X$ induces for the vertices corresponding to $L(T_i)$.
  Note that $T_i$ and $T'_i$ enforce the same restrictions on these orders.
  Since the construction ensures that vertices on adjacent levels that correspond to the same element of $X$ are adjacent in $G$, the vertices are ordered consistently between adjacent levels by \Cref{lem:levelPlanarEmbedding}, which corresponds to Property 2 of \Cref{lem:pqprobCharacterization}.
  Therefore, $\mathcal I$ and $\mathcal I'$ are equivalent.
\end{proof}

\begin{lemma}
  \label{lem:tlp2pqprob}
  There exists a linear-time reduction from \tlp to \pqprob.
\end{lemma}
\begin{proof}
  Let $\mathcal I = (G, \ell, (C_1, \dots, C_k))$ be an instance of \tlp. 
  Recall that we can assume that each laminar family $C_i$ is represented by a tree with leaves $\ell^{-1}(i)$.
  Due to \Cref{lem:level-monotone}, we can additionally assume that~$G$ is a matching. We construct an equivalent instance $\mathcal I' = (X, (T_1, \dots, T_k))$ of \pqprob as follows.
  The elements of $X$ correspond bijectively to the edges of $G$.
  For each $i \in [k]$, the PQ-tree $T_i$ represents the consecutivity constraints obtained from $C_i$ by replacing each vertex with the unique edge it is incident to. As before, the linear order of the vertices on each level $i$ in a solution to $\mathcal I$ corresponds to a satisfying order of the corresponding PQ-tree $T_i$ in $\mathcal I'$.
  For the converse direction, we also need to ensure level planarity, that is the characterization of \Cref{lem:levelPlanarEmbedding}.
  As each edge of $\mathcal I$ between levels $i$ and $i+1$ is represented by a unique element of $X$ in $\mathcal I'$, this is ensured by point 2 of \Cref{lem:pqprobCharacterization}.
\end{proof}

This shows that both problems are linear-time equivalent.
Recall that \tlp admits a quadratic-time algorithm that uses a chain of linear-time reductions via the problems \textsc{Connected SEFE-2} and 2-fixed \simpq \cite{albfr-ibpcl-14,albfr-ibpcl-15}, where for the latter an explicit quadratic solution exists~\cite{br-spqoa-16}.
Especially reducing to \textsc{Connected SEFE-2} loses much of the structure contained in the instance, which is why we want to show that the quadratic algorithm for 2-fixed \simpq can also be directly applied to our problems.

The problem \simpq, introduced by Bläsius and Rutter~\cite{br-spqoa-16}, takes as input a DAG whose vertices are PQ-trees and where every arc with source $T_a$ and target $T_b$ injectively maps $L(T_b)$ to $L(T_a)$.
The question is whether there exist linear orders for the leaves of all PQ-trees such that for every arc, the chosen order for the source induces a suborder of the chosen order for the target.
While this problem is NP-complete in the general case, Bläsius and Rutter showed that the problem can be solved in quadratic time if the instance is what they call \emph{2-fixed}~\cite{br-spqoa-16}.
In the following, we show that we can model an instance of \tlp as an equivalent instance of \simpq where the DAG has maximum degree 2, which also implies 2-fixedness~\cite{br-spqoa-16}.

\begin{theorem}\label{thm:pqrob2simpq}
  There exists a linear-time reduction from \pqprob to 2-fixed \simpq.
  The resulting instance can be solved in quadratic time.
\end{theorem}
\begin{proof}
  Let $\mathcal I = (X, (T_1, \dots, T_k))$ be an instance of \pqprob.
  The DAG in the equivalent instance of \simpq simply consists of the PQ-trees $T_1, \dots T_k$, together with a set of $k-1$ further PQ-trees that ensure that, between each pair of consecutive PQ-trees in $T_1, \dots T_k$, their common elements are ordered consistently.
  More specifically, for each $i \in [k - 1]$, we add a PQ-tree $S_i$ consisting of a single inner P-node with $|L(T_i) \cap L(T_{i+1})|$ leaves as children and add arcs from $S_i$ to $T_i$ and $T_{i+1}$ that map each pair of common elements in $T_i$ and $T_{i+1}$ to a distinct leaf of $S_i$.
  Thus, the DAG is a path with edges of alternating orientation, and the vertices thus alternate between being sinks ($T_1, \dots T_k$) and sources ($S_1, \dots S_{k-1}$).
  The equivalence of the instances is immediate by \Cref{lem:pqprobCharacterization}.
  Moreover, since every vertex of the DAG has degree at most 2, the instance is 2-fixed and can be solved in quadratic time~\cite{br-spqoa-16}.
\end{proof}

\section{Conclusion}
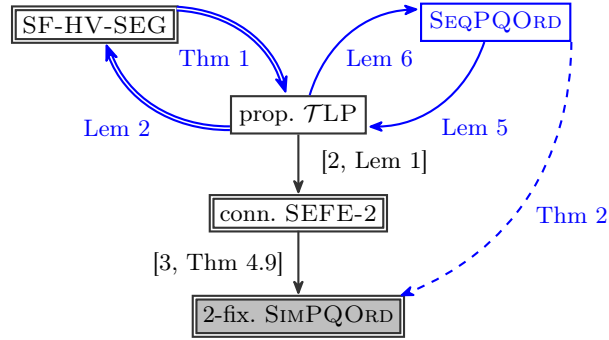
\begin{figure}[t]
  \centering
  \usetikzlibrary{arrows.meta,shapes.misc,quotes}
  \begin{tikzpicture}[>={Stealth[round]},thick,black!80,text=black,
                        def/.style={draw, rectangle}, new/.style={blue}, algo/.style={fill=lightgray},
                        quadratic/.style={double}, trivial/.style={"trivial",dashed}]

      \node (tree)   [def]                                           {prop.\ $\mathcal{T}$\textsc{LP}};
      \node (hv)     [def, above left=1 of tree, quadratic]          {\hvp};
      \node (seqpq)  [def, above right=1 of tree, new]               {\textsc{SeqPQOrd}};
      \node (csefe2) [def, below=1 of tree.base, quadratic]          {conn.\ \textsc{SEFE-2}};
      \node (simpq)  [def, below=1 of csefe2.base, quadratic, algo]  {2-fix.\ \textsc{SimPQOrd}};
  
      \draw[->, shorten >=1pt, bend left=35,pos=0.45]
          (hv) edge["{Thm \ref{thm:hv-to-rtlp}}",pos=0.6,auto=right,inner sep=2pt,new,quadratic] (tree)
          (tree) edge["{Lem \ref{lem:tlp-hv}}",new, quadratic] (hv);
  
      \draw[->, shorten >=1pt, bend left=35,pos=0.56]
          (seqpq) edge["{Lem \ref{lem:pqprob2tlp}}",new] (tree)
          (tree) edge["{Lem \ref{lem:tlp2pqprob}}",pos=0.4,auto=right,inner sep=2pt,new] (seqpq);

      \draw[->, shorten >=1pt]
          (tree) to["{\cite[Lem 1]{albfr-ibpcl-15}}",pos=0.45, inner sep=5pt] (csefe2);
      \draw[->, shorten >=1pt, bend left=34,new, dashed]
          (seqpq.south east) to["{Thm \ref{thm:pqrob2simpq}}"] (simpq);
      \draw[->, shorten >=1pt]
          (csefe2) to["{\cite[Thm 4.9]{br-spqoa-16}}", auto=right, inner sep=5pt] (simpq);
    \end{tikzpicture}
  \caption{
Relation of the problems we consider, with new problems and relations shown in blue.
    Dashed, single- and double-line arrows represent trivial, linear- and quadratic-time reductions, respectively.
    Double-border boxes can be solved (possibly via reductions) in quartic time, single-border boxes in quadratic time.
    Shaded boxes have an explicit solution not using reductions.
  }
  \label{fig:schema}
\end{figure}

In this paper, we resolve the remaining open question from the 1989 seminal works by Kratochv{\'\i}l and Matou{\v s}ek on segment intersection graphs \cite{km-igo-94,kra-asp-94,nf-ppa-92}.
When posing the question, the authors already suggested two reformulations of the problem in terms of forbidden volkswagen patterns in 2-layer drawings or forbidden cross patterns in the adjacency matrix.
Still, it took more than 30 years until another different perspective finally allowed us to show tractability using the constrained planar graph drawing problem \tlp.
Interestingly, \tlp itself is also a reformulation of another problem called \textsc{$k$-ary Tanglegrams} \cite{wsp-gkatl-12,albfr-ibpcl-15}, which arose in computational biology to visualize evolutionary histories of species.
In addition to this very visual solution, we also give another equivalent characterization as a constrained ordering problem that allows a direct solution without involved reductions.
\Cref{fig:schema} gives a graphical overview over the reductions and equivalences thus shown in this work.
Through the links established in our paper, we now know that we can view the same tractable problem through 6 different-yet-equivalent formulations, some of which arose naturally and independently, and taking different perspectives from either a graph drawing or entirely combinatorial standpoint.
It will be interesting to see whether further problems turn out to join this order of equivalent formulations.

\appendix
\section{Acknowledgements}
We want to thank Jan Kratochv{\'\i}l for introducing us to the ``Volkswagen-Problem'' at HOMONOLO'22 and, together with Jiří Fiala, also for organizing this very enjoyable workshop.
We also want to thank Marie Diana Sieper and Felix Klesen who joined us at the workshop for developing early approaches towards the \hvp problem.

\section{Funding}
Simon D.\ Fink was funded by the Vienna Science and Technology Fund (WWTF) [10.47379/ICT22029].
Matthias Pfretzschner was funded by the Deutsche Forschungsgemeinschaft (DFG, German Research Foundation) -- 541433306.
Peter Stumpf was funded by Czech Science Foundation grant no.
23-04949X.

\bibliographystyle{splncs04}
\bibliography{references}

\end{document}